\documentclass[a4paper,11pt]{article}

\usepackage{epsfig}
\usepackage{subfigure}
\usepackage{calc}
\usepackage{amssymb}
\usepackage{amstext}
\usepackage{amsmath}
\usepackage{amsthm}
\usepackage{multicol}
\usepackage{pslatex}
\usepackage{fullpage}

\usepackage[small]{caption}

\usepackage{amsfonts}
\usepackage{dsfont} %Pour mathds
\usepackage{stmaryrd}
\usepackage[english]{varioref}

\newtheorem{definition}{Definition}
\newtheorem{proposition}{Proposition}
\newtheorem{theorem}{Theorem}
\newtheorem{example}{Example}
\newtheorem{remark}{Remark}

\begin{document}

%\title{\uppercase{Authors' Instructions}  \subtitle{Preparation of Camera-Ready
% Contributions to SciTePress Proceedings} }
% \title{\uppercase{Chaotic iterations for steganography} 
% \subtitle{Stego-security and topological-security}
% }
\title{\uppercase{Chaotic iterations for steganography}\\
Stego-security and topological-security}

%\author{\authorname{First Author Name\sup{1}, Second Author Name\sup{1} and
%Third Author Name\sup{2}}
%\author{\authorname{Nicolas Friot, Christophe Guyeux, and Jacques
%M. Bahi\sup{1}}
%\affiliation{\sup{1}Institute of Problem Solving, XYZ University, My Street,
%MyTown, MyCountry}
%\affiliation{\sup{2}Department of Computing, Main University, MySecondTown,
%MyCountry}
%\affiliation{\sup{1}Computer Science Laboratory LIFC,University of
%Franche-Comt\'e, 16 route de Gray, Besan\c con, France}
%\email{\{f\_author, s\_author\}@ips.xyz.edu, t\_author@dc.mu.edu}
%\email{\{jacques.bahi, nicolas.friot,  and
%christophe.guyeux\}@lifc.univ-fcomte.fr}}

\author{ Nicolas Friot, Christophe Guyeux, and Jacques M. Bahi\\
\\
Computer Science Laboratory LIFC\\
University of Franche-Comt\'{e}\\
16 route de Gray, Besan\c con, France\\
\\
\{nicolas.friot, christophe.guyeux, jacques.bahi\}@lifc.univ-fcomte.fr\\
}
%\keywords{The paper must have at least one keyword. The text must be set to
%9-point font size and without the use of bold or italic font style. For more
%than one keyword, please use a semicolon as a separator. Keywords must be
%titlecased.}
%\keywords{Steganography; Chaos; Security; Chaotic Iterations}

\maketitle

Chaotic Iterations.

%\abstract{The abstract should summarize the contents of the paper and should
%contain at least 70 and at most 200 words. The text must be set to 9-point font
%size.}

\abstract{In this paper is proposed a novel steganographic scheme based on
chaotic iterations. This research work takes place into the information hiding
security fields. We show that the proposed scheme is stego-secure, which is the
highest level of security in a well defined and studied category of attack
called ``watermark-only attack''. Additionally, we prove that this scheme
presents topological properties so that it is one of the firsts able to face, at
least partially, an adversary when considering the others categories of attacks
defined in the literature. }

%\onecolumn \maketitle \normalsize \vfill

\section{Introduction}
\label{sec:introduction}

\noindent Robustness and security are two major concerns in information hiding
\cite{DBLP:conf/sswmc/KatzenbeisserD04,DBLP:conf/mdai/Domingo-FerrerB08}.
%Even if security and robustness are neighboring concepts without clearly established definitions~\cite{Perez-Freire06}, robustness is often considered to be mostly concerned with blind elementary attacks, whereas security is not limited to certain specific attacks.
%Indeed, security encompasses robustness and intentional attacks~\cite{Kalker2001,ComesanaPP05bis}. 
%The best attempt to give an elegant and concise definition for each of these two terms was proposed in \cite{Kalker2001}.
%Following Kalker, we will consider in this research work that 
These two concerns have been defined in \cite{Kalker2001} as follows. ``Robust
watermarking is a mechanism to create a communication channel that is
multiplexed into original content [...]. It is required that, firstly, the
perceptual degradation of the marked content [...] is minimal and, secondly,
that the capacity of the watermark channel degrades as a smooth function of the
degradation of the marked content. [...]. Watermarking security refers to the
inability by unauthorized users to have access to the raw watermarking channel.
[...] to remove, detect and estimate, write or modify the raw watermarking
bits.'' We will focus in this research work on security.

In the framework of watermarking and steganography, security has seen several important developments since the last decade~\cite{BarniBF03,Cayre2005,Ker06,DBLP:journals/tdp/Bras-AmorosD08}. 
The first fundamental work in security was made by Cachin in the context of steganography~\cite{Cachin2004}. 
Cachin interprets the attempts of an attacker to distinguish between an innocent image and a stego-content as a hypothesis testing problem. 
In this document, the basic properties of a stegosystem are defined using the notions of entropy, mutual information, and relative entropy. 
Mittelholzer, inspired by the work of Cachin, proposed the first theoretical framework for analyzing the security of a watermarking scheme~\cite{Mittelholzer99}.

These efforts to bring a theoretical framework for security in steganography and watermarking have been followed up by Kalker, who tries to clarify the concepts (robustness \emph{vs.} security), and the classifications of watermarking attacks~\cite{Kalker2001}. 
This work has been deepened by Furon \emph{et al.}, who have translated Kerckhoffs' principle (Alice and Bob shall only rely on some previously shared secret for privacy), from cryptography to data hiding~\cite{Furon2002}. 
They used Diffie and Hellman methodology, and Shannon's cryptographic framework~\cite{Shannon49}, to classify the watermarking attacks into categories, according to the type of information Eve has access to~\cite{Cayre2005,Perez06}, namely: Watermarked Only Attack (WOA), Known Message Attack (KMA), Known Original Attack (KOA), and Constant-Message Attack (CMA).
Levels of security have been recently defined in these setups.
The highest level of security in WOA is called stego-security \cite{Cayre2008},
whereas topological-security tends to improve the ability to withstand attacks
in KMA, KOA, and CMA setups \cite{gfb10:ip}.

To the best of our knowledge, there exist only two information hiding schemes
that are both stego-secure and topologically-secure \cite{gfb10:ip}. The first
one is based on a spread spectrum technique called Natural Watermarking. It is stego-secure when its parameter $\eta$ is equal to $1$ \cite{Cayre2008}. 
Unfortunately, this scheme is neither robust, nor able to face an attacker in KOA and KMA setups, due to its lack of a topological property called expansivity \cite{gfb10:ip}.
The second scheme both topologically-secure and stego-secure is based on chaotic
iterations \cite{guyeux10ter}. However, it allows to embed securely only one bit per embedding parameters. 
The objective of this research work is to improve the scheme presented by authors of \cite{guyeux10ter}, in such a way that more than one bit can be embedded.

The remainder of this document is organized as follows.
In Section \ref{sec:Basic Recalls}, some basic recalls concerning both chaotic iterations and Devaney's chaos are given.
In Section \ref{section:IH based on CIs} are presented results and information hiding scheme on which our work is based.
Classes of attacks considered in this paper are detailed in Section \ref{section:IH security}.
Stego-security and topological-security are recalled too in this section.
The new information hiding scheme is given in Section \ref{section:Algorithm}. 
Its stego-security is studied in the next section.
The topological framework making it possible to evaluate topological-security is
introduced in Section \ref{sec:SCISMM-topological-model}. Then the topological
properties of our scheme are investigated in the next section, leading to the
evaluation of its topological-security.
%Finally, consequences of the security study are outlined in Section \ref{sec:consequences}.
This research work ends by a conclusion section where our contribution is summarized and intended future researches are presented.

\section{Basic Recalls}
\label{sec:Basic Recalls}

\subsection{Chaotic Iterations}
\label{sec:chaotic iterations}

In the sequel $S^{n}$ denotes the $n^{th}$ term of a sequence $S$ and
$V_{i}$ is for the $i^{th}$ component of a vector $V$. Finally, the following notation
is used: $\llbracket0;N\rrbracket=\{0,1,\hdots,N\}$.\newline

Let us consider a \emph{system} of a finite number $\mathsf{N}$ of elements (or
\emph{cells}), so that each cell has a boolean \emph{state}. A sequence of length
$\mathsf{N}$ of boolean states of the cells corresponds to a particular
\emph{state of the system}. A sequence that elements belong into $\llbracket
0;\mathsf{N-1} \rrbracket $ is called a \emph{strategy}. The set of all
strategies is denoted by $\mathbb{S}.$

\begin{definition}
\label{Def:chaotic iterations}
The set $\mathds{B}$ denoting $\{0,1\}$, let
$f:\mathds{B}^{\mathsf{N}}\longrightarrow \mathds{B}^{\mathsf{N}}$ be a function
and $S\in \mathbb{S}$ be a strategy.
The so-called \emph{chaotic iterations} are
defined by $x^0\in \mathds{B}^{\mathsf{N}}$ and $\forall (n,i) \in
\mathds{N}^{\ast} \times \llbracket0;\mathsf{N-1}\rrbracket$:
\begin{equation*}
%\forall n\in \mathds{N}^{\ast }, \forall i\in \llbracket1;\mathsf{N}\rrbracket
%,x_i^n=\left\{
x_i^n=\left\{
\begin{array}{ll}
x_i^{n-1} & \text{ if }S^n\neq i, \\
\left(f(x^{n-1})\right)_{S^n} & \text{ if }S^n=i.\end{array}\right.\end{equation*}
\end{definition}

\subsection{Devaney's Chaotic Dynamical Systems}
\label{sec:Devaney}

Some topological definitions and properties taken from the
mathematical theory of chaos are recalled in this section.
%It is important to understand these notions which will be used in
%the rest of this article especially in the Section~\ref{sec:SCISMM-topological-model}~\vpageref{sec:SCISMM-topological-model} and the Section~\ref{sec:SCISMM-chaos-security}~\vpageref{sec:SCISMM-chaos-security}.

%\subsubsection{Notations}

%\begin{itemize}
%  \item Let $(\mathcal{X},d)$ be a metric space.
%  \item Let $f$ be a continuous function on  $(\mathcal{X},d)$.
%\end{itemize}

Let $(\mathcal{X},d)$ be a metric space and $f$ a continuous function on  $(\mathcal{X},d)$.

%\subsubsection{Topological definitions used in chaos}

\begin{definition}%[Topological transitivity]
 $f$ is said to be \emph{topologically transitive} if, for any pair
of open sets $U,V\subset \mathcal{X}$, there exists $k>0$ such that $f^{k}(U)\cap
V\neq\varnothing $.
\end{definition}

\begin{definition}%[Regularity]
$(\mathcal{X},f)$ is said to be \emph{regular} if the set of
periodic points is dense in $\mathcal{X}$.
\end{definition}

\begin{definition}%[Sensitivity]
$f$ has \emph{sensitive dependence on initial conditions} if there exists $\delta >0$ such that, for any $x\in
\mathcal{X}$ and any neighborhood $V$ of $x$, there exist $y\in V$ and
$n\geqslant 0$ such that $d\left(f^{n}(x),f^{n}(y) \right)>\delta $.

$\delta $ is called the \emph{constant of sensitivity} of $f$.
\end{definition}

It is now possible to introduce the well-established mathematical definition of
chaos~\cite{Devaney89},

\begin{definition}%[Chaotic function]
A function $f:\mathcal{X}\longrightarrow \mathcal{X}$ is said to be
\emph{chaotic} on $\mathcal{X}$ if:
\begin{enumerate}
  \item $f$ is regular,
  \item $f$ is topologically transitive,
  \item $f$ has sensitive dependence on initial conditions.
\end{enumerate}
\end{definition}

When $f$ is chaotic, then the system $(\mathcal{X}, f)$ is chaotic and quoting
Devaney: ``it is unpredictable because of the sensitive dependence on initial
conditions. It cannot be broken down or simplified into two subsystems which do
not interact because of topological transitivity. And in the midst of this
random behavior, we nevertheless have an element of regularity''. Fundamentally
different behaviors are consequently possible and occur in an unpredictable
way.

Let us finally remark that,

\begin{theorem}[\cite{Banks92}]\label{theo:banks}
If a function is regular and topologicaly transitive on a metric space, then the
function is sensitive on initial conditions.
\end{theorem}

%The main interest of the Banks theorem is to reduce the properties to check when
%we want to prove that a function on a metric space is chaotic. This theorm will
%be used in the
%Section~\ref{sec:SCISMM-chaos-security}~\vpageref{sec:SCISMM-chaos-security}.

\section{Information hiding based on chaotic iterations}
\label{section:IH based on CIs}

\subsection{Topology of Chaotic Iterations}\label{sec:topological}

\noindent In this section, we give the outline proofs establishing the
topological properties of chaotic iterations. As our scheme is inspired by the work of Guyeux \emph{et al.} \cite{gfb10:ip,guyeux10ter,bg10:ij}, the proofs detailed at the end of this document will follow a same canvas.

Let us firstly introduce some notations and terminologies.

\begin{definition}%[Strategy adapter]
\label{def:strategy-adapter}
Let $\mathsf{k} \in \mathds{N}^\ast$. 
A \emph{strategy adapter} is a sequence which elements belong into $\llbracket 0, \mathsf{k-1} \rrbracket$.
The set of all strategies with terms in $\llbracket 0, \mathsf{k-1} \rrbracket$
is denoted by  $\mathbb{S}_\mathsf{k}$.
\end{definition}

\begin{definition}%[Discrete boolean metric]
\label{def:discret-boolean-metric}
The \emph{discrete boolean metric} is the
application $\delta: \mathds{B} \longrightarrow \mathds{B}$ defined by
$\delta(x,y)=0\Leftrightarrow x=y.$
\end{definition}

\begin{definition}%[Initial function]
Let $k \in \mathds{N}^\ast$. 
The \emph{initial function} is the map $i_k$ defined by:
\begin{equation*}
\begin{array}{cccc}
i_k: & \mathbb{S}_k & \longrightarrow & \llbracket 0, \mathsf{k-1} \rrbracket \\
   & (S^{n})_{n\in \mathds{N}} & \longmapsto & S^{0} \\
\end{array}
\end{equation*}
\end{definition}

\begin{definition}%[Shift function]
Let $k \in \mathds{N}^\ast$. 
The \emph{shift function} is the map $\sigma_k$ defined by:
\begin{equation*}
\begin{array}{cccc}
\sigma_k : & \mathbb{S}_k & \longrightarrow & \mathbb{S}_k\\
         & (S^{n})_{n\in \mathds{N}} & \longmapsto & (S^{n+1})_{n\in \mathds{N}}
\end{array}
\end{equation*}
\end{definition}

\begin{definition}%[Function $F_{f}$]
Given a function $f:\mathds{B}^{\mathsf{N}} \rightarrow \mathds{B}^{\mathsf{N}}$, the function \emph{$F_{f}$} is defined by:
\begin{equation*}
\begin{array}{ll}
F_{f}: & \llbracket 0;\mathsf{N-1}\rrbracket \times
\mathds{B}^{\mathsf{N}} \longrightarrow \mathds{B}^{\mathsf{N}} \\
 & (k,E) \longmapsto \left( E_{j}.\delta (k,j)+f(E)_{k}.\overline{\delta
(k,j)}\right)_{j\in \llbracket 0;\mathsf{N-1}\rrbracket} \\
\end{array}
\end{equation*}
\end{definition}

\begin{definition}%[Phase space]
The phase space used for chaotic iterations is denoted by $\mathcal{X}_1$ and
defined by $\mathcal{X}_1=\mathbb{S}_\mathsf{N} \times
\mathds{B}^{\mathsf{N}}.$
\end{definition}

\begin{definition}%[Map $G_{f}$]
Given a function $f:\mathds{B}^{\mathsf{N}} \rightarrow \mathds{B}^{\mathsf{N}}$, the map \emph{$G_{f}$} is defined by:
\begin{equation*}
\begin{array}{cccc}
G_{f}: & \mathcal{X}_1 & \longrightarrow & \mathcal{X}_1 \\
 & \left( S,E\right) & \longmapsto &\left( \sigma_\mathsf{N} (S),F_{f}(i_\mathsf{N}(S),E)\right)
 \\
\end{array}
\end{equation*}
\end{definition}

With these definitions, chaotic iterations can be described by the following iterations of the
discret dynamical system:
%\begin{equation*}
%\left\{
%\begin{array}{l}
%X^{0}\in \mathcal{X} \\
%X^{k+1}=G_{f}(X^{k}).
%\end{array}
%\right.
%\end{equation*}

\begin{equation*}
\left\{
\begin{array}{l}
X^{0}\in \mathcal{X}_1\\
\forall k \in \mathds{N}^\ast, X^{k+1}=G_{f}(X^{k})
\end{array}
\right.
\end{equation*}

%We have defined in~\cite{bg10:ij} a new distance $d_1$ between two points by:
Finally, a new distance $d_1$ between two points has been defined by:

\begin{definition}[Distance $d_1$ on  $\mathcal{X}_1$]\label{def:distance-on-X1}
$\forall (S,E),(\check{S},\check{E} )\in \mathcal{X}_1,$
$d_1((S,E);(\check{S},\check{E}))=d_{\mathds{B}^{\mathsf{N}}}(E,\check{E})+d_{\mathbb{S}_\mathsf{N}}(S,\check{S}),$
where:
\begin{itemize}
\item
$\displaystyle{d_{\mathds{B}^{\mathsf{N}}}(E,\check{E})}=\displaystyle{\sum_{k=0}^{\mathsf{N-1}}\delta
(E_{k},\check{E}_{k})} \in \llbracket 0 ; \mathsf{N} \rrbracket$
\item
$\displaystyle{d_{\mathbb{S}_\mathsf{N}}(S,\check{S})}=\displaystyle{\dfrac{9}{\mathsf{N}}\sum_{k=1}^{\infty
}\dfrac{|S^{k}-\check{S}^{k}|}{10^{k}}} \in [0 ; 1].$
\end{itemize}
are respectively two distances on $\mathds{B}^{\mathsf{N}}$ and $\mathbb{S}_\mathsf{N}$
($\forall \mathsf{N} \in \mathds{N}^\ast$).
\end{definition}

\begin{remark}
This new distance has been introduced by authors of \cite{bg10:ij}
 to satisfy the following
requirements. When the number of different cells between two systems is
increasing, then their distance should increase too. In addition, if two systems
present the same cells and their respective strategies start with the same
terms, then the distance between these two points must be small, because the
evolution of the two systems will be the same for a while. The distance presented
above follows these recommendations. 
%Indeed, if the floor value $\lfloor
%d(X,Y)\rfloor $ is equal to $n$, then the systems $E, \check{E}$ differ in $n$
%cells. In addition, $d(X,Y) - \lfloor d(X,Y) \rfloor $ is a measure of the
%differences between strategies $S$ and $\check{S}$. More precisely, this floating
%part is less than $10^{-k}$ if and only if the first $k$ terms of the two
%strategies are equal. Moreover, if the $k^{th}$ digit is nonzero, then the
%$k^{th}$ terms of the two strategies are different.
\end{remark}

It is then proven that,

\begin{proposition}\label{Prop:continuite}
$G_f$ is a continuous function on $(\mathcal{X}_1,d_1)$, for all $f:\mathds{B}^\mathsf{N} \rightarrow \mathds{B}^\mathsf{N}$.
\end{proposition}

Let us now recall the iteration function used by authors of \cite{guyeux10ter}.

\begin{definition}%[Vectorial negation]
The \emph{vectorial negation} is the function defined by:
\begin{equation*}
\begin{array}{cccc}
f_{0} : & \mathbb{B}^\mathsf{N} & \longrightarrow & \mathbb{B}^\mathsf{N}\\
& (b_0,\cdots,b_\mathsf{N-1}) & \longmapsto &
(\overline{b_0},\cdots,\overline{b_\mathsf{N-1}})\\
\end{array}
\end{equation*}
\end{definition}
%\begin{definition}[Vectorial Negation]\label{Def:vectorial-negation}
%\begin{equation}
%\begin{array}{cccc}
%f_0\ :\ & \mathbb{B}^N & \longrightarrow & \mathbb{B}^N \\
%& (b_1,\cdots,b_\mathsf{N}) & \longmapsto &
%(\overline{b_1},\cdots,\overline{b_\mathsf{N}})\\
%\end{array}
%\end{equation}
%\end{definition}
%It is then checked in \cite{bg10:ij} that i
In the metric space $(\mathcal{X}_1,d_1)$, $G_{f_0}$  satisfies the three
conditions for Devaney's chaos: regularity, transitivity, and
sensitivity. %~\cite{bg10:ij}. 
So,

%the following result.
\begin{theorem}%[$G_{f_0}$ is a chaotic map]
$G_{f_0}$ is a chaotic map on $(\mathcal{X}_1,d_1)$ according to Devaney.
\end{theorem}

Finally, it has been stated in ~\cite{bg10:ij} that,

\begin{proposition}%[Cardinality of $\mathcal{X}_1$]
\label{prop:cardinality-X1}
The phase space $\mathcal{X}_1$ has, at least, the cardinality of the 
continuum.
\end{proposition}

\subsection{Chaotic Iterations for Data Hiding}

To explain how to use chaotic iterations for information hiding, we
must firstly define the significance of a given coefficient.

\subsubsection{Most and Least Significant
Coefficients}\label{sec:msc-lsc}

We first notice that terms of the original content $x$ that may be replaced by terms issued
from the watermark $y$ are less important than other: they could be changed 
without be perceived as such. More generally, a 
\emph{signification function} 
attaches a weight to each term defining a digital media,
depending on its position $t$.

\begin{definition}%[Signification function]
A \emph{signification function} is a real sequence 
$(u^k)^{k \in \mathds{N}}$. % with a limit equal to 0.
\end{definition}

\begin{example}\label{Exemple LSC}
Let us consider a set of    
grayscale images stored into portable graymap format (P3-PGM):
each pixel ranges between 256 gray levels, \textit{i.e.},
is memorized with eight bits.
In that context, we consider 
$u^k = 8 - (k  \mod  8)$  to be the $k$-th term of a signification function 
$(u^k)^{k \in \mathds{N}}$. 
Intuitively, in each group of eight bits (\textit{i.e.}, for each pixel) 
the first bit has an importance equal to 8, whereas the last bit has an
importance equal to 1. This is compliant with the idea that
changing the first bit affects more the image than changing the last one.
\end{example}

\begin{definition}%[Significance of coefficients]
\label{def:msc,lsc}
Let $(u^k)^{k \in \mathds{N}}$ be a signification function, 
$m$ and $M$ be two reals s.t. $m < M$. 
\begin{itemize}
\item The \emph{most significant coefficients (MSCs)} of $x$ is the finite 
  vector  $$u_M = \left( k ~ \big|~ k \in \mathds{N} \textrm{ and } u^k 
    \geqslant M \textrm{ and }  k \le \mid x \mid \right);$$
 \item The \emph{least significant coefficients (LSCs)} of $x$ is the 
finite vector 
$$u_m = \left( k ~ \big|~ k \in \mathds{N} \textrm{ and } u^k 
  \le m \textrm{ and }  k \le \mid x \mid \right);$$
 \item The \emph{passive coefficients} of $x$ is the finite vector 
   $$u_p = \left( k ~ \big|~ k \in \mathds{N} \textrm{ and } 
u^k \in ]m;M[ \textrm{ and }  k \le \mid x \mid \right).$$
 \end{itemize}
 \end{definition}

For a given host content $x$,
MSCs are then ranks of $x$  that describe the relevant part
of the image, whereas LSCs translate its less significant parts.
These two definitions are illustrated on Figure~\ref{fig:MSCLSC}, where the significance function $(u^k)$ is defined as in Example \ref{Exemple LSC}, $M=5$, and $m=6$.

\begin{figure}[htb]

\begin{minipage}[b]{.31\linewidth}
  \centering
  \centerline{\includegraphics[width=3.cm]{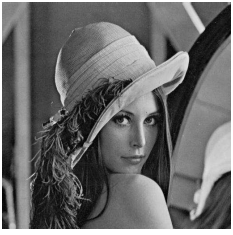}}
   %\centerline{\epsfig{figure=img/lena512.pdf,width=4cm}}
  \centerline{(a) Original Lena.}
\end{minipage}
\hfill
\begin{minipage}[b]{.31\linewidth}
  \centering
    \centerline{\includegraphics[width=3.cm]{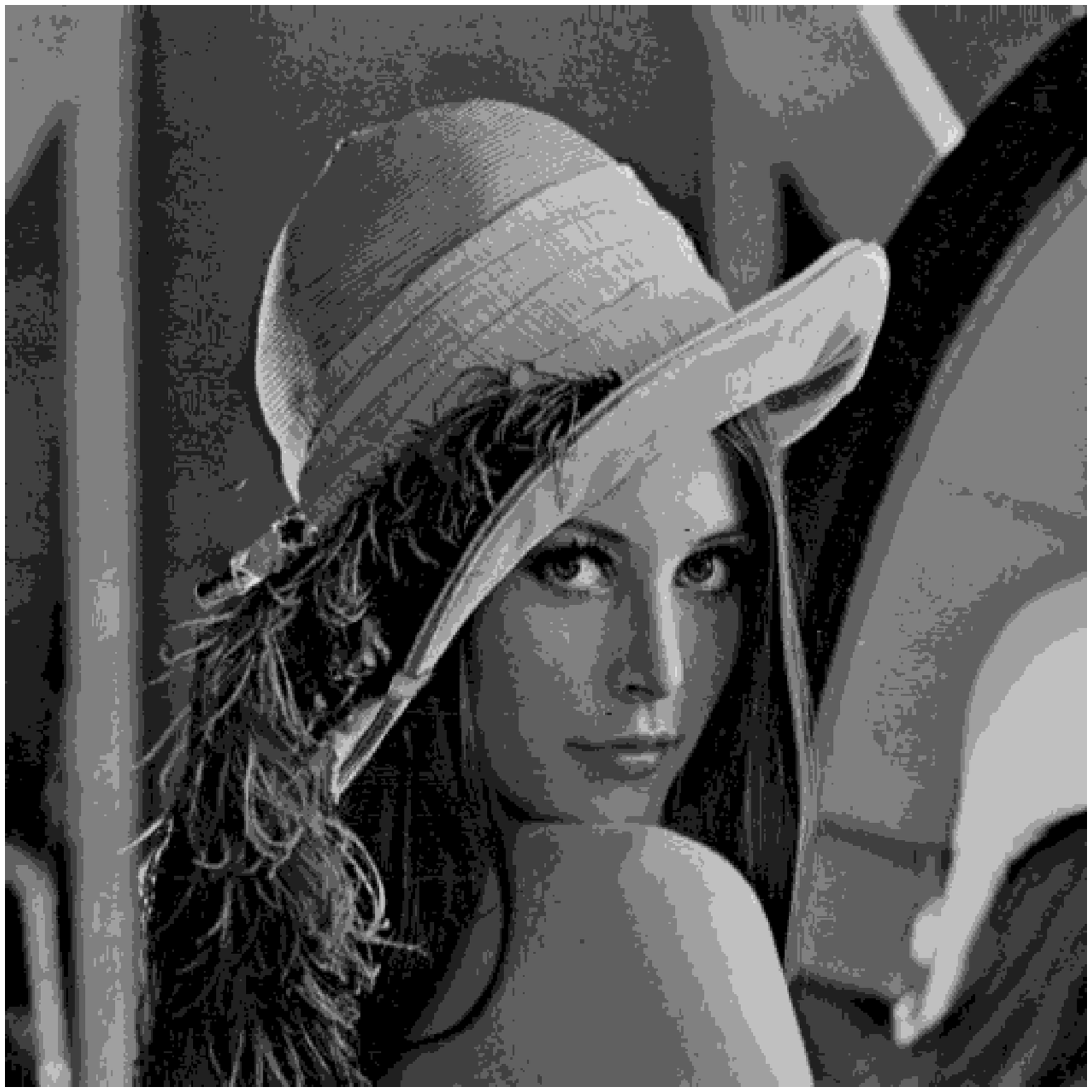}}
   %\centerline{\epsfig{figure=img/lena_msb_678.pdf,width=4cm}}
  \centerline{(b) MSCs of Lena.}
\end{minipage}
\hfill
\begin{minipage}[b]{0.31\linewidth}
  \centering
    \centerline{\includegraphics[width=3.cm]{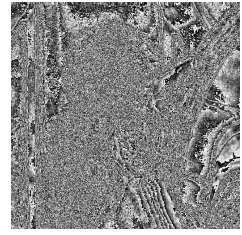}}
  
 %\centerline{\epsfig{figure=img/lena_lsb_1234_facteur17.pdf,width=4cm}}
  \centerline{(c) LSCs of Lena ($\times 17$).}
\end{minipage}
\caption{Most and least significant coefficients of Lena.}
\label{fig:MSCLSC}
\end{figure}

\subsubsection{Presentation of the Scheme}

Authors of \cite{guyeux10ter} have proposed to use chaotic iterations as an information hiding scheme, as follows. 
Let:

\begin{itemize}
  \item $(K,N) \in [0;1]\times \mathds{N}$ be an embedding key,
  \item $X \in \mathbb{B}^\mathsf{N}$ be the $\mathsf{N}$ LSCs of a cover $C$,% $X$ be the initial state $X_0$,
  \item $(S^n)_{n \in \mathds{N}} \in \llbracket 0, \mathsf{N-1}
  \rrbracket^{\mathds{N}}$ be a strategy, which depends on the message to hide $M \in [0;1]$ and $K$,
  \item $f_0 : \mathbb{B}^\mathsf{N} \rightarrow \mathbb{B}^\mathsf{N}$ be the vectorial logical negation.
\end{itemize}

So the watermarked media is $C$ whose LSCs are replaced by $Y_K=X^{N}$, where:

\begin{equation*}
\left\{
 \begin{array}{l}
X^0 = X\\
\forall n < N, X^{n+1} = G_{f_0}\left(X^n\right).\\
\end{array} \right.
\end{equation*}

Two ways to generate $(S^n)_{n \in \mathds{N}}$ are given by these authors, namely
Chaotic Iterations with Independent Strategy~(CIIS) and Chaotic Iterations with Dependent
Strategy~(CIDS). 
In CIIS, the strategy is independent from the cover media $C$, whereas in CIDS the strategy will be dependent on $C$. 
As we will use the CIIS strategy in this document, we recall it below.
Finally, MSCs are not used here, as we do not consider the case of authenticated watermarking.

\subsubsection{CIIS Strategy}

Let us firstly give the definition of the Piecewise Linear Chaotic Map~(PLCM, see~\cite{Shujun1}):

%\begin{definition}%[PLCM]
%The \emph{Piecewise Linear Chaotic Map} is defined by
\begin{equation*}
F(x,p)=\left\{
 \begin{array}{ccc}
x/p & \text{if} & x \in [0;p], \\
(x-p)/(\frac{1}{2} - p) & \text{if} & x \in \left[ p; \frac{1}{2} \right],
\\
F(1-x,p) & \text{else,} & \\
\end{array} \right.
\end{equation*}
%\end{definition}

\noindent where $p \in \left] 0; \frac{1}{2} \right[$ is a ``control parameter''.

Then, the general term of the strategy $(S^n)_n$ in CIIS setup is defined by
the following expression: $S^n = \left \lfloor \mathsf{N} \times K^n \right \rfloor +
1$, where:

\begin{equation*}
\left\{
 \begin{array}{l}
p \in \left[ 0 ; \frac{1}{2} \right] \\
K^0 = M \otimes K\\
K^{n+1} = F(K^n,p), \forall n \leq N_0\\ \end{array} \right.
\end{equation*}

\noindent in which $\otimes$ denotes the bitwise exclusive or (XOR) between two floating part numbers (\emph{i.e.}, between their binary digits representation).

\section{Data hiding security}
\label{section:IH security}

\subsection{Classification of Attacks}\label{sec:attack-classes}

In the steganography framework, attacks have been classified
in~\cite{Cayre2008} as follows.

\begin{definition}
Watermark-Only Attack (WOA) occurs when an attacker has only access to several watermarked contents.
\end{definition}

\begin{definition}
Known-Message Attack (KMA) occurs when an attacker has access to several pairs of watermarked contents and
corresponding hidden messages.
\end{definition}

\begin{definition}
Known-Original Attack (KOA) is when an attacker has access to several pairs of watermarked contents and
their corresponding original versions.
\end{definition}

\begin{definition}
Constant-Message Attack (CMA) occurs when the attacker observes several watermarked contents and only knows that the
unknown hidden message is the same in all contents.
\end{definition}

\subsection{Stego-Security}\label{sec:stego-security}

In the prisoner problem of Simmons~\cite{Simmons83,DBLP:conf/ih/BergmairK06}, Alice and Bob are in jail, and they want to, possibly, devise an escape plan by exchanging hidden messages in innocent-looking cover contents. 
These messages are to be conveyed to one another by a common warden, Eve, who over-drops all contents and can choose to interrupt the communication if they appear to be stego-contents. 

The stego-security, defined in this
framework, is the highest security level in WOA setup~\cite{Cayre2008}.
To recall it, we need the following notations:
\begin{itemize}
  \item $\mathds{K}$ is the set of embedding keys,
  \item $p(X)$ is the probabilistic model of $N_0$ initial host contents,
  \item $p(Y|K_1)$ is the probabilistic model of $N_0$ watermarked contents.
\end{itemize}

Furthermore, it is supposed in this context that each host content has been watermarked with the same secret key $K_1$ and the same embedding function $e$.

It is now possible to define the notion of stego-security:

\begin{definition}[Stego-Security]
\label{Def:Stego-security}
The embedding function $e$ is \emph{stego-secure} if and only if:
$$\forall K_1 \in \mathds{K}, p(Y|K_1)=p(X).$$
\end{definition}

To the best of our knowledge, until now, only two schemes have been proven to be
stego-secure. On the one hand, the authors of \cite{Cayre2008} have established
that the spread spectrum technique called Natural Watermarking is stego-secure
when its distortion parameter $\eta$ is equal to $1$. On the other hand, it has
been proven in~\cite{gfb10:ip} that:

\begin{proposition}
Chaotic Iterations with Independent Strategy (CIIS)  are stego-secure.
\label{prop:CIIS-stego-security}
\end{proposition}

\subsection{Topological-Security}\label{sec:chaos-security}

To check whether an information hiding scheme $S$ is topologically-secure or
not, $S$ must be written as an iterate process $x^{n+1}=f(x^n)$ on a metric space
$(\mathcal{X},d)$. This formulation is always possible~\cite{ih10}. So,

\begin{definition}[Topological-Security]
\label{Def:chaos-security-definition}
An information hiding scheme $S$ is said to be topologically-secure on
$(\mathcal{X},d)$ if its iterative process has a chaotic
behavior according to Devaney.
\end{definition}

In the approach presented by Guyeux \emph{et al.}, a data hiding scheme is secure if it is unpredictable. 
Its iterative process must satisfy the Devaney's chaos property and its level of
topological-security increases with the number of chaotic properties satisfied
by it.

This new concept of security for data hiding schemes has been proposed in~\cite{ih10} as a complementary approach to the existing framework. 
It contributes to the reinforcement of confidence into existing secure data hiding schemes. 
Additionally, the study of security in KMA, KOA, and CMA setups is realizable in this context. 
Finally, this framework can replace stego-security in situations that are not encompassed by it. 
In particular, this framework is more relevant to give evaluation of data hiding schemes claimed as chaotic.

\section{The improved algorithm}
\label{section:Algorithm}

\noindent In this section is introduced a new algorithm that generalize the
scheme presented by authors of \cite{guyeux10ter}.

Let us firstly introduce the following notations:

\begin{itemize}
  %\item Let $\mathbb{S}_k=\llbracket 0, \mathsf{k-1} \rrbracket^{\mathds{N}}$ be
  %the set of all strategies with values in $\llbracket 0, \mathsf{k-1}
  %\rrbracket$.
  \item $x^0 \in \mathbb{B}^\mathsf{N}$ is the $\mathsf{N}$ least
  significant coefficients of a given cover media $C$.% $X$ be the initial state $X_0$,
  %\item The vectorial negation $f_0$ (defined in~\ref{Def:vectorial-negation})
  %as iteration function.
  \item $m^0 \in \mathbb{B}^\mathsf{P}$ is the watermark to embed into $x^0$.
  \item $S_p \in \mathbb{S}_N$ is a strategy called \textbf{place strategy}.
  \item $S_c \in \mathbb{S}_P$ is a strategy called \textbf{choice strategy}.
  \item Lastly, $S_m \in \mathbb{S}_P$ is a strategy called \textbf{mixing strategy}.
\end{itemize}

%\subsection{Algorithm SCISMM in details}\label{sec:algo-SCISMM}

Our information hiding scheme called Steganography by Chaotic Iterations and
Substitution with Mixing Message (SCISMM) is defined by
%\begin{definition}[SCISMM]\label{def:algo}
$\forall (n,i,j) \in
\mathds{N}^{\ast} \times \llbracket 0;\mathsf{N-1}\rrbracket \times \llbracket
0;\mathsf{P-1}\rrbracket$:

\begin{equation*}
\left\{
\begin{array}{l}
x_i^n=\left\{
\begin{array}{ll}
x_i^{n-1} & \text{ if }S_p^n\neq i \\
m_{S_c^n} & \text{ if }S_p^n=i.
\end{array}
\right.
\\
\\
m_j^n=\left\{
\begin{array}{ll}
m_j^{n-1} & \text{ if }S_m^n\neq j \\
 & \\
\overline{m_j^{n-1}} & \text{ if }S_m^n=j.
\end{array}
\right.
\end{array}
\right.
\end{equation*}
%\end{definition}
where $\overline{m_j^{n-1}}$ is the boolean negation of $m_j^{n-1}$.
%\subsection{Stego-content by SCISMM}\label{sec:stego-content}

The stego-content is the boolean vector $y = x^P \in
\mathbb{B}^\mathsf{N}$.
  
\section{Study of stego-security}

\noindent Let us prove that,

\begin{proposition}
SCISMM is stego-secure.
\end{proposition}

\begin{proof}
%We suppose that $K,M \sim
%\mathbf{U}\left([0;1]\right)$, and $X \sim
%We consider a strategy in the CIIS scheme, and we
Let us suppose that $x^0 \sim
\mathbf{U}\left(\mathbb{B}^N\right)$ and $m^0 \sim
\mathbf{U}\left(\mathbb{B}^P\right)$ in a SCISMM setup. 
We will prove by a mathematical induction that $\forall n \in \mathds{N}, x^n
\sim \mathbf{U}\left(\mathbb{B}^N\right)$. The base case is obvious according
to the uniform repartition hypothesis. 

Let us now suppose that the statement $x^n \sim
\mathbf{U}\left(\mathbb{B}^N\right)$ holds for some $n$. 
For a given $k  \in \mathbb{B}^N$,
we denote by $\tilde{k_i} \in \mathbb{B}^N$ the vector defined by: $\forall i
\in  \llbracket 0;\mathsf{N-1}\rrbracket,$
if
$k=\left(k_0,k_1,\hdots,k_i,\hdots,k_{\mathsf{N}-2},k_{\mathsf{N}-1}\right)$,\newline
then $\tilde{k}_i=\left( k_0,k_1,\hdots,\overline{k_i},\hdots,k_{\mathsf{N}-2},k_{\mathsf{N}-1} \right).$

%\begin{equation*}
Let
$E_{i,j}$ be the following events: $$\forall (i,j) \in \llbracket
0;\mathsf{N-1}\rrbracket \times \llbracket 0;\mathsf{P-1}\rrbracket, E_{i,j}=$$
$$S_p^{n+1}=i \wedge S_c^{n+1}=j \wedge  m_j^{n+1}=k_i \wedge \left(x^n=k \vee x^n=\tilde{k}_i\right),$$
and
$p=P\left(x^{n+1}=k\right).$
So,
\begin{equation*}
p=P\left(\bigvee_{i \in \llbracket
0;\mathsf{N-1}\rrbracket, j \in \llbracket 0;\mathsf{P-1}\rrbracket}
E_{i,j} \right).
\end{equation*}
We now introduce the following notation: $P_1(i)=P\left(S_p^{n+1}=i\right),$
 $P_2(j)=P\left(S_c^{n+1}=j\right),$
 $P_3(i,j)=P\left(m_j^{n+1}=k_i\right),$
and $P_4(i)=P\left(x^n=k \vee x^n=\tilde{k}_i\right).$

These four events are independent in SCISMM setup, thus:
\begin{equation*}
p=\sum_{i \in \llbracket
0;\mathsf{N-1}\rrbracket, j \in \llbracket 0;\mathsf{P-1}\rrbracket}
P_1(i)P_2(i)P_3(i,j)P_4(i).
\end{equation*}
According to 
Proposition~\ref{prop:CIIS-stego-security}, $P\left(m_j^{n+1}=k_i\right)=\frac{1}{2}$.
As the two events are incompatible:
$$P\left(x^n=k \vee x^n=\tilde{k}_i\right)=P\left(x^n=k\right) + P\left(x^n=\tilde{k}_i\right).$$

Then, by using the inductive hypothesis:
$P\left(x^n=k\right)=\frac{1}{2^N},$
and
$P\left(x^n=\tilde{k}_i\right)=\frac{1}{2^N}.$

Let $S$ be defined by $$S=\sum_{i
\in \llbracket 0;\mathsf{N-1}\rrbracket, j \in \llbracket 0;\mathsf{P-1}\rrbracket}
P_1(i)P_2(j).$$
Then $p =  2 \times \frac{1}{2} \times \frac{1}{2^N} \times S = \frac{1}{2^N} \times S$.

$S$ can now be evaluated:
\begin{equation*}
\begin{array}{lll}
S & = &  \sum_{i
\in \llbracket 0;\mathsf{N-1}\rrbracket, j \in \llbracket 0;\mathsf{P-1}\rrbracket}
P_1(i)P_2(j)\\
  & = &  \sum_{i \in \llbracket 0;\mathsf{N-1}\rrbracket} P_1(i) \times \sum_{j \in \llbracket 0;\mathsf{P-1}\rrbracket}
P_2(j).\\
\end{array}
\end{equation*}

The set of events $\left \{ S_p^{n+1}=i \right \}$ for $i \in \llbracket 0;N-1
\rrbracket$ and the set of events $\left \{ S_c^{n+1}=j \right \}$ for $j \in
\llbracket 0;P-1 \rrbracket$ are both a partition of the universe of possible,
so $S=1$.

%Evaluate now $ P\left(S^n=k\right)$.
%\newline
%We have: $K^0 = M \otimes K \sim
%\mathbf{U}\left([0;1]\right)$, so $K^\prime=F\left( K^0,p \right) \sim
%\mathbf{U}\left([0;1]\right)$, because for PLCMs, ``Uniform input generate
%uniform output''~\cite{Lasota}.
%Well with an immediate proof by recurrence we have:
%$K^n \sim \mathbf{U}\left([0;1]\right)$. As $S^n = \left \lfloor N \times X^n
%\right \rfloor + 1$ we can deduce that: $S^n \sim
%\mathbf{U}\left([1;N]\right)$. Therefore: $ P\left(S^n=k\right)=\frac{1}{N}$.
Finally, 
%\begin{equation*}
%\begin{array}{llr}
%P\left(X^{n+1}=e\right) & =\frac{1}{2^N} \sum_{k=1}^N
% P\left(S^n=k\right) & \\
% & =\frac{1}{2^N} \sum_{k=1}^N \frac{1}{N} & \\
% & =\frac{1}{2^N} & \\
%\end{array}
%\end{equation*}
%\begin{equation*}
%$\begin{array}{llr}
%P\left(X^{n+1}=e\right) & =\frac{1}{2^N} \sum_{k=1}^N
% P\left(S^n=k\right) & \\
%P\left(X^{n+1}=e\right) & =\frac{1}{2^N} \sum_{k=1}^N \frac{1}{N} &
% =\frac{1}{2^N} \\
%\end{array}$.
$P\left(x^{n+1}=k\right)=\frac{1}{2^N}$, which leads to $x^{n+1} \sim
\mathbf{U}\left(\mathbb{B}^N\right)$.
This result is true $\forall n \in \mathds{N}$, we thus have proven that the
stego-content $y$ is uniform in the set of possible stego-content, so
$y \sim \mathbf{U}\left(\mathbb{B}^N\right) \text{
when } x \sim \mathbf{U}\left(\mathbb{B}^N\right).$
\end{proof}

\section{Topological model}
\label{sec:SCISMM-topological-model}

%\section{THE NEW TOPOLOGICAL SPACE}
\noindent In this section, we prove that SCISMM can be modeled as a
discret dynamical system in a topological space. 
We will show in the next section that SCISMM is a case of topological chaos in the sense of Devaney.

\subsection{Iteration Function and Phase Space}

\label{Defining}

%\begin{definition}[Function $F$]
Let
\begin{equation*}
\begin{array}{ll}
F: & \llbracket0;\mathsf{N-1}\rrbracket \times
\mathds{B}^{\mathsf{N}}\times \llbracket0;\mathsf{P-1}\rrbracket  \times
\mathds{B}^{\mathsf{P}}
\longrightarrow \mathds{B}^{\mathsf{N}} \\
 & (k,x,\lambda,m)  \longmapsto  \left( \delta
 (k,j).x_j+\overline{\delta (k,j)}.m_{\lambda}\right) _{j\in
 \llbracket0;\mathsf{N-1}\rrbracket}\\
 \end{array}%
\end{equation*}%
%\end{definition}

\noindent where + and . are the boolean addition and product operations.

Consider the phase space $\mathcal{X}_2$ defined as follow:

%\begin{definition}[Phase space $\mathcal{X}_2$]

\begin{equation*}
\mathcal{X}_2 =  \mathbb{S}_N \times \mathds{B}^\mathsf{N} \times \mathbb{S}_P
\times \mathds{B}^\mathsf{P} \times \mathbb{S}_P,
\end{equation*}
where $\mathbb{S}_N$ and $\mathbb{S}_P$ are the sets introduced in Section \ref{section:Algorithm}.
%\end{definition}

We define the map $\mathcal{G}_{f_0}:\mathcal{X}_2  \longrightarrow  \mathcal{X}_2$ by:

\begin{equation*}
\mathcal{G}_{f_0}\left(S_p,x,S_c,m,S_m\right) =   
\end{equation*}
\begin{equation*}
\left(\sigma_N(S_p),
F(i_N(S_p),x,i_P(S_c),m),\sigma_P(S_c),G_{f_0}(m,S_m),\sigma_P(S_m)\right)
\end{equation*}

\noindent Then SCISMM can be described by the
iterations of the following discret dynamical system:

%\begin{definition}[Discret Dynamical System for SCISMM]
\begin{equation*}
\left\{
\begin{array}{l}
X^0 \in \mathcal{X}_2 \\
X^{k+1}=\mathcal{G}_{f_0}(X^k).%
\end{array}%
\right.
\end{equation*}%
%\end{definition}%

\subsection{Cardinality of $\mathcal{X}_2$}

By comparing $\mathcal{X}_2$ and $\mathcal{X}_1$, we have the following result.

\begin{proposition}
The phase space $\mathcal{X}_2$ has, at least, the cardinality of the continuum.
\end{proposition}

\begin{proof}
Let $\varphi$ be the map defined as follow:

\begin{equation*}
\begin{array}{lcll}
\varphi: & \mathcal{X}_1 & \longrightarrow & \mathcal{X}_2 \\
 & (S,x) &  \longmapsto  &(S,x,0,0,0)\\
 \end{array}%
\end{equation*}%

\noindent $\varphi$ is  injective.
So the cardinality of  $\mathcal{X}_2$ is greater than or equal to the
cardinality of $\mathcal{X}_1$. And consequently
$\mathcal{X}_2$ has at least the cardinality of the continuum.
\end{proof}

\begin{remark}
This result is independent on the number of cells of the system.
\end{remark}

\subsection{A New Distance on $\mathcal{X}_2$}

We define a new distance on $\mathcal{X}_2$ as follow: 
%\begin{definition}[Distance $d_2$ on $\mathcal{X}_2$]
$\forall X,\check{X} \in \mathcal{X}_2,$ if $X =(S_p,x,S_c,m,S_m)$ and $\check{X} =
(\check{S_p},\check{x},\check{S_c},\check{m}, \check{S_m}),$ then:
\begin{equation*}
\begin{array}{lll}
d_2(X,\check{X}) & = &
d_{\mathds{B}^{\mathsf{N}}}(x,\check{x})+d_{\mathds{B}^{\mathsf{P}}}(m,\check{m})\\
& + &
d_{\mathbb{S}_N}(S_p,\check{S_p})+d_{\mathbb{S}_P}(S_c,\check{S_c})+d_{\mathbb{S}_P}(S_m,\check{S_m}),\\
\end{array}
\end{equation*}
where $d_{\mathds{B}^{\mathsf{N}}}$, $d_{\mathds{B}^{\mathsf{P}}}$,
$d_{\mathbb{S}_N}$, and $d_{\mathbb{S}_P}$ are the same distances than in
Definition~\ref{def:distance-on-X1}.
%\end{definition}

\subsection{Continuity of SCISMM}

To prove that SCISMM is another example of topological chaos in the
sense of Devaney, $\mathcal{G}_{f_0}$ must be continuous on the metric
space $(\mathcal{X}_2,d_2)$.

\begin{proposition}
$\mathcal{G}_{f_0}$ is a continuous function on $(\mathcal{X}_2,d_2)$.
\end{proposition}

\begin{proof}
We use the sequential continuity.

Let $((S_p)^n,x^n,(S_c)^n,m^n,(S_m)^n)_{n\in \mathds{N}}$ be a sequence of the
phase space $\mathcal{X}_2$, which converges to $(S_p,x,S_c,m,S_m)$. We will prove
that $\left( \mathcal{G}_{f_0}((S_p)^n,x^n,(S_c)^n,m^n,(S_m)^n)\right) _{n\in
\mathds{N}}$ converges to $\mathcal{G}_{f_0}(S_p,x,S_c,m,S_m)$. Let us recall
that for all $n$, $(S_p)^n$, $(S_c)^n$ and $(S_m)^n$ are strategies, thus we
consider a sequence of strategies (\emph{i.e.}, a sequence of sequences).

As $d_2(((S_p)^n,x^n,(S_c)^n,m^n,(S_m)^n),(S_p,x,S_c,m,S_m))$ converges to 0,
each distance
$d_{\mathds{B}^{\mathsf{N}}}(x^n,x)$, $d_{\mathds{B}^{\mathsf{P}}}(m^n,m)$,  
$d_{\mathbb{S}_N}((S_p)^n,S_p)$, $d_{\mathbb{S}_P}((S_c)^n,S_c)$, and
$d_{\mathbb{S}_P}((S_m)^n,S_m)$ converges to 0. But
$d_{\mathds{B}^{\mathsf{N}}}(x^n,x)$ and $d_{\mathds{B}^{\mathsf{P}}}(m^n,m)$
are integers, so $\exists n_{0}\in \mathds{N}, \forall
n \geqslant n_0,
d_{\mathds{B}^{\mathsf{N}}}(x^n,x)=0$ and $\exists n_{1}\in \mathds{N},
\forall n \geqslant n_1, d_{\mathds{B}^{\mathsf{P}}}(m^n,m)=0$. 

Let $n_3=Max(n_0,n_1)$.
In other words, there exists a threshold $n_{3}\in \mathds{N}$ after which no
cell will change its state:
$\exists n_{3}\in \mathds{N},n\geqslant n_{3}\Longrightarrow (x^n = x) \wedge
(m^n = m).$

In addition, $d_{\mathbb{S}_N}((S_p)^n,S_p)\longrightarrow 0$,
$d_{\mathbb{S}_P}((S_c)^n,S_c)\longrightarrow 0$, and
$d_{\mathbb{S}_P}((S_m)^n,S_m)\longrightarrow 0$, so $\exists n_4,n_5,n_6\in
 \mathds{N},$
 \begin{itemize}
   \item $\forall n \geqslant n_4, d_{\mathbb{S}_N}((S_p)^n,S_p)<10^{-1}$,
   \item $\forall n \geqslant n_5, d_{\mathbb{S}_P}((S_c)^n,S_c)<10^{-1}$,
   \item $\forall n \geqslant n_6, d_{\mathbb{S}_P}((S_m)^n,S_m)<10^{-1}$.
 \end{itemize}
 
 Let $n_7=Max(n_4,n_5,n_6)$. For
 $n\geqslant n_7$, all the strategies $(S_p)^n$, $(S_c)^n$, and $(S_m)^n$ have the same
 first term, which are respectively $(S_p)_0$,$(S_c)_0$ and $(S_m)_0$ :$\forall n\geqslant n_7,$
$$((S_p)_0^n={(S_p)}_0)\wedge
((S_c)_0^n={(S_c)}_0)\wedge ((S_m)_0^n={(S_m)}_0).
$$

Let $n_8=Max(n_3,n_7)$.
After the $n_8-$th term, states of $x^n$ and $x$ on the one hand, and $m^n$ and $m$ on the other hand, are
identical. Additionally, strategies $(S_p)^n$ and $S_p$, $(S_c)^n$ and $S_c$, and $(S_m)^n$
and $S_m$  start with the same first term.

Consequently, states of
$\mathcal{G}_{f_0}((S_p)^n,x^n,(S_c)^n,m^n,(S_m)^n)$ and $\mathcal{G}_{f_0}(S_p,x,S_c,m,S_m)$ are equal,
so, after the $(n_8)^{th}$ term, the distance $d_2$ between these two
points is strictly smaller than $3.10^{-1}$, so strictly smaller than 1.

 We now prove that the distance between $\left(
\mathcal{G}_{f_0}((S_p)^n,x^n,(S_c)^n,m^n,(S_m)^n)\right) $ and $\left(
\mathcal{G}_{f_0}(S_p,x,S_c,m,S_m)\right) $ is convergent to 0.
Let $\varepsilon >0$. 

\begin{itemize}
\item If $\varepsilon \geqslant 1$, we have seen that distance
between $\mathcal{G}_{f_0}((S_p)^n,x^n,(S_c)^n,m^n,(S_m)^n)$ and 
$\mathcal{G}_{f_0}(S_p,x,S_c,m,S_m)$ is
strictly less than 1 after the $(n_8)^{th}$ term (same state).
\medskip

\item If $\varepsilon <1$, then $\exists k\in \mathds{N},10^{-k}\geqslant
\frac{\varepsilon}{3} \geqslant 10^{-(k+1)}$. As
$d_{\mathbb{S}_N}((S_p)^n,S_p)$, $d_{\mathbb{S}_P}((S_c)^n,S_c)$ and
$d_{\mathbb{S}_P}((S_m)^n,S_m)$  converges to 0, we have:
\begin{itemize}
 \item $\exists n_9\in \mathds{N},\forall n\geqslant
 n_9,d_{\mathbb{S}_N}((S_p)^n,S_p)<10^{-(k+2)}$,
 \item $\exists n_{10}\in \mathds{N},\forall n\geqslant
 n_{10},d_{\mathbb{S}_P}((S_c)^n,S_c)<10^{-(k+2)}$,
 \item $\exists n_{11}\in \mathds{N},\forall n\geqslant
 n_{11},d_{\mathbb{S}_P}((S_m)^n,S_m)<10^{-(k+2)}$.
\end{itemize}
  
Let $n_{12}=Max(n_9,n_{10},n_{11})$
thus after $n_{12}$, the $k+2$ first terms of $(S_p)^n$ and $S_p$, $(S_c)^n$ and
$S_c$, and $(S_m)^n$ and $S_m$, are equal.
\end{itemize}

\noindent As a consequence, the $k+1$ first entries of the strategies of
$\mathcal{G}_{f_0}((S_p)^n,x^n,(S_c)^n,m^n,(S_m)^n)$ and $
\mathcal{G}_{f_0}(S_p,x,S_c,m,S_m)$ are the same
(due to the shift of strategies) and following the definition of
$d_{\mathbb{S}_N}$ and $d_{\mathbb{S}_P}$:
$$d_2\left(\mathcal{G}_{f_0}((S_p)^n,x^n,(S_c)^n,m^n,(S_m)^n);\mathcal{G}_{f_0}(S_p,x,S_c,m,S_m)\right)$$
is equal to :
$$d_{\mathbb{S}_N}((S_p)^n,S_p) + d_{\mathbb{S}_P}((S_c)^n,S_c) +
d_{\mathbb{S}_P}((S_m)^n,S_m)$$
which is smaller than $3.10^{-(k+1)} \leqslant 3.\frac{\varepsilon}{3} =\varepsilon$.

Let $N_{0}=max(n_{8},n_{12})$. We can claim that
\begin{equation*}
\forall \varepsilon >0,\exists N_{0}\in \mathds{N},\forall n\geqslant N_{0},
\end{equation*}
\begin{equation*}
d_2\left(\mathcal{G}_{f_0}((S_p)^n,x^n,(S_c)^n,m^n,(S_m)^n);\mathcal{G}_{f_0}(S_p,x,S_c,m,S_m)\right)
\leqslant \varepsilon .
\end{equation*}
$\mathcal{G}_{f_0}$ is consequently continuous on $(\mathcal{X}_2,d_2)$.
\end{proof}

% >>>>>>>>>>>>>>>>>>>>>> Discrete chaotic iterations as topological chaos <<<<<<<<<<<<<<<<<<<<<<<<<<<<<<
\section{SCISMM is chaotic}
\label{section:chaos-security}

\noindent To prove that we are in the framework of Devaney's topological chaos,
we have to check the regularity, transitivity, and sensitivity conditions.
%It will be proven that $\mathcal{G}_{f_0}$ function as introduce in
%definition~\ref{} satisfies these three hypotheses.

\subsection{Regularity}

\begin{proposition}%[Regularity of $\mathcal{G}_{f_0}$]
\label{prop:regularite}
Periodic points of $\mathcal{G}_{f_0}$ are dense in $\mathcal{X}_2$.
\end{proposition}

\begin{proof}
Let $(\check{S_p},\check{x},\check{S_c},\check{m}, \check{S_m})\in
\mathcal{X}_2$ and $\varepsilon >0$. We are looking for a periodic point
$(\widetilde{S_p},\widetilde{x},\widetilde{S_c},\widetilde{m}, \widetilde{S_m})$
satisfying $d_2((\check{S_p},\check{x},\check{S_c},\check{m},
\check{S_m});(\widetilde{S_p},\widetilde{x},\widetilde{S_c},\widetilde{m}, \widetilde{S_m}))<\varepsilon$.

As $\varepsilon$ can be strictly lesser than 1, we must choose
$\widetilde{x} = \check{x}$ and
$\widetilde{m} = \check{m}$.
Let us define $k_0(\varepsilon) =\lfloor
- log_{10}(\frac{\varepsilon}{3} )\rfloor +1$ and consider the set:
$\mathcal{S}_{\check{S_p},\check{S_c},\check{S_m}, k_0(\varepsilon)} =$
$\left\{ S
\in \mathbb{S}_N \times \mathbb{S}_P \times \mathbb{S}_P/ ((S_p)^k =
\check{S_p}^k) \wedge ((S_c)^k =\check{S_c}^k))\right.$
$\left.\wedge ((S_m)^k=\check{S_m}^k)), \forall k \leqslant k_0(\varepsilon)
\right\}$.

Then, 
$\forall (S_p,S_c,S_m) \in
\mathcal{S}_{\check{S_p},\check{S_c},\check{S_m},k_0(\varepsilon)},$
$d_2((S_p,\check{x},S_c,\check{m},S_m);(\check{S_p},\check{x},\check{S_c},\check{m},
\check{S_m})) < 3.\frac{\varepsilon}{3}=\varepsilon.$
It remains to choose $(\widetilde{S_p},\widetilde{S_p},\widetilde{S_p}) \in
\mathcal{S}_{\check{S_p},\check{S_c},\check{S_m}, k_0(\varepsilon)}$ such that
$(\widetilde{S_p},\widetilde{x},\widetilde{S_c},\widetilde{m}, \widetilde{S_m})
= (\widetilde{S_p},\check{x},\widetilde{S_c},\check{m}, \widetilde{S_m})$ is
a periodic point for $\mathcal{G}_{f_0}$.

Let $\mathcal{J} =$
$\left\{ i \in \llbracket0;\mathsf{N-1}\rrbracket / x_i \neq
\check{x_i}, \text{ where }\right.$
$\left.(S_p,x,S_c,m,S_m) = \mathcal{G}_{f_0}^{k_0}
(\check{S_p},\check{x},\check{S_c},\check{m},\check{S_m}) \right\},$
$\lambda = card(\mathcal{J})$, and
$j_0 <j_1 < ... < j_{\lambda-1}$ the elements of $ \mathcal{J}$. 

\begin{enumerate}
  \item Let us firstly build three strategies: $S_p^\ast$,  $S_c^\ast$, and 
  $S_m^\ast$, as follows.
\begin{enumerate}
  \item 
$(S_p^\ast)^k=\check{S_p}^k$, $(S_c^\ast)^k=\check{S_c}^k$, and
$(S_m^\ast)^k=\check{S_m}^k$, if $k \leqslant k_0(\varepsilon).$ 
\item Let us now explain how to replace $\check{x}_{j_q}$, $\forall q \in
  \llbracket0;\mathsf{\lambda-1}\rrbracket$:\newline
First of all, we must replace $\check{x}_{j_0}$:
\begin{enumerate}
  \item If $\exists \lambda_0 \in \llbracket0;\mathsf{P-1}\rrbracket /
  \check{x}_{j_0}=m_{\lambda_0}$, then we can choose
 $(S_p^\ast)^{k_0+1}=j_0$, $(S_c^\ast)^{k_0+1}=\lambda_0$,
 $(S_m^\ast)^{k_0+1}=\lambda_0$, and so $I_{j_0}$ will be equal to $1$.
 \item If such a $\lambda_0$ does not exist, we choose:\newline $(S_p^\ast)^{k_0+1}=j_0$, $(S_c^\ast)^{k_0+1}=0$,
 $(S_m^\ast)^{k_0+1}=0$,\newline $(S_p^\ast)^{k_0+2}=j_0$,
 $(S_c^\ast)^{k_0+2}=0$, $(S_m^\ast)^{k_0+2}=0$, \newline and 
 $I_{j_0}=2$.\newline

All of the $\check{x}_{j_q}$ are replaced similarly. 
The other terms of $S_p^\ast$,  $S_c^\ast$, and 
  $S_m^\ast$ are constructed identically, and the values of $I_{j_q}$ are defined in
  the same way.\newline
 Let $\gamma = \sum_{q=0}^{\lambda-1}I_{j_q}$.
\end{enumerate}

\item Finally, let $(S_p^\ast)^{k}=(S_p^\ast)^{j}$, $(S_c^\ast)^{k}=(S_c^\ast)^{j}$,
and $(S_m^\ast)^{k}=(S_m^\ast)^{j}$, where $j\leqslant k_{0}(\varepsilon )+\gamma$
is satisfying $j\equiv k~[\text{mod } (k_{0}(\varepsilon )+\gamma)]$, if
$k>k_{0}(\varepsilon )+\gamma$.
\end{enumerate}
So,
$\mathcal{G}_{f_0}^{k_{0}(\varepsilon
)+\gamma}(S_p^\ast,\check{x},S_c^\ast,\check{m},S_m^\ast)=(S_p^\ast,\check{x},S_c^\ast,m,S_m^\ast).$
Let $\mathcal{K} =\left\{ i \in \llbracket0;\mathsf{P-1}\rrbracket / m_i \neq
\check{m_i}, \text{ where }\right.$\newline
$\left.\mathcal{G}_{f_0}^{k_{0}(\varepsilon
)+\gamma}(S_p^\ast,\check{x},S_c^\ast,\check{m},S_m^\ast)=(S_p^\ast,\check{x},S_c^\ast,m,S_m^\ast)
\right\},$\newline
$\mu = card(\mathcal{K})$, and
$r_0 <r_1 < ... < r_{\mu-1}$ the elements of $ \mathcal{K}$.

  \item Let us now build the strategies $\widetilde{S_p}$, 
  $\widetilde{S_c}$, $\widetilde{S_m}$.
\begin{enumerate}
  \item Firstly, let $\widetilde{S_p}^k=(S_p^\ast)^k$, $\widetilde{S_c}^k=(S_c^\ast)^k$, and
$\widetilde{S_m}^k=(S_m^\ast)^k$, if $k \leqslant k_0(\varepsilon) + \gamma.$
\item  How to replace $\check{m}_{r_q}, \forall q \in
  \llbracket0;\mathsf{\mu-1}\rrbracket$:\newline
  
First of all, let us explain how to replace $\check{m}_{r_0}$:
\begin{enumerate}
  \item If $\exists \mu_0 \in \llbracket0;\mathsf{N-1}\rrbracket /
  \check{x}_{\mu_0}=m_{r_0}$, then 
 we can choose $\widetilde{S_p}^{k_0+\gamma +1}=\mu_0$, $\widetilde{S_c}^{k_0+\gamma
 +1}=r_0$, $\widetilde{S_m}^{k_0+\gamma
 +1}=r_0$.\newline In that situation, we define $J_{r_0}=1$.
 \item If such a $\mu_0$ does not exist, then we can choose:\newline
$\widetilde{S_p}^{k_0+\gamma +1}=0$, $\widetilde{S_c}^{k_0+\gamma
 +1}=r_0$, $\widetilde{S_m}^{k_0+\gamma
 +1}=r_0$,\newline 
$\widetilde{S_p}^{k_0+\gamma +2}=0$, $\widetilde{S_c}^{k_0+\gamma
 +2}=r_0$, $\widetilde{S_m}^{k_0+\gamma
 +2}=0$,\newline 
$\widetilde{S_p}^{k_0+\gamma +3}=0$, $\widetilde{S_c}^{k_0+\gamma
 +3}=r_0$, $\widetilde{S_m}^{k_0+\gamma
 +3}=0$.\newline 
 Let $J_{r_0}=3$.

Then the other $\check{m}_{r_q}$ are replaced as previously,
  the other terms of $\widetilde{S_p}$, 
  $\widetilde{S_c}$, and $\widetilde{S_m}$ are constructed in the same way, and the
  values of $J_{r_q}$ are defined similarly.\newline
 Let $\alpha = \sum_{q=0}^{\mu-1}J_{r_q}$.
\end{enumerate}
\item Finally, let $\widetilde{S_p}^{k}=\widetilde{S_p}^{j}$,
$\widetilde{S_c}^{k}=\widetilde{S_c}^{j}$, and
$\widetilde{S_m}^{k}=\widetilde{S_m}^{j}$ where $j\leqslant k_{0}(\varepsilon
)+\gamma +\alpha$ is satisfying $j\equiv k~[\text{mod } (k_{0}(\varepsilon
)+\gamma + \alpha)]$, if $k>k_{0}(\varepsilon )+\gamma + \alpha$.
\end{enumerate}
\end{enumerate}

So, $\mathcal{G}_{f_0}^{k_{0}(\varepsilon
)+\gamma+\alpha}(\widetilde{S_p},\check{x},\widetilde{S_c},\check{m},\widetilde{S_m})=(\widetilde{S_p},\check{x},\widetilde{S_c},\check{m},\widetilde{S_m})$
%\end{enumerate}

Then, $(\widetilde{S_p},\widetilde{S_c},\widetilde{S_m}) \in
\mathcal{S}_{\check{S_p},\check{S_c},\check{S_m},k_0(\varepsilon)}$ defined as
previous is such that $(\widetilde{S_m},\check{x},\widetilde{S_m},\check{m},\widetilde{S_m})$
 is a periodic point, of period $k_{0}(\varepsilon
)+\gamma+\alpha$, which is $\varepsilon -$close to $(\check{S_p},\check{x},\check{S_c},\check{m}, \check{S_m})$.

As a conclusion, $(\mathcal{X}_2,\mathcal{G}_{f_0})$ is regular.
\end{proof}

\subsection{Transitivity}\label{sec:transitivite}

\begin{proposition}%[Transitivity of $ \mathcal{G}_{f_0}$]
\label{prop:transitivite}
$(\mathcal{X}_2,\mathcal{G}_{f_0})$ is topologically transitive.
\end{proposition}

\begin{proof}
Let us define $\mathcal{X}:\mathcal{X}_2\rightarrow \mathbb{B}^{\mathsf{N}},$
such that $\mathcal{X}(S_p,x,S_c,m,S_m)=x$ and
$\mathcal{M}:\mathcal{X}_2\rightarrow \mathbb{B}^{\mathsf{P}},$ such that
$\mathcal{M}(S_p,x,S_c,m,S_m)=m$. Let \linebreak
$\mathcal{B}_{A}=\mathcal{B}(X_{A},r_{A}) $ and
$\mathcal{B}_{B}=\mathcal{B}(X_{B},r_{B})$ be two open balls of $\mathcal{X}_2$,
with\linebreak $X_{A}=((S_p)_{A},x_{A},(S_c)_{A},m_{A},(S_m)_{A})$ and
$X_{B}=((S_p)_{B},x_{B},(S_c)_{B},m_{B},(S_m)_{B})$. We are looking for
$\widetilde{X}=(\widetilde{S_p},\widetilde{x},\widetilde{S_c},\widetilde{m},\widetilde{S_m})$
in $\mathcal{B}_{A} $ such that $\exists n_{0}\in
\mathbb{N},\mathcal{G}_{f_0}^{n_{0}}(\widetilde{X})\in \mathcal{B}_{B}$.\newline
$\widetilde{X}$ must be in $\mathcal{B}_{A}$ and $r_{A}$ can be strictly
lesser than 1, so $\widetilde{x}=x_{A}$ and $\widetilde{m}=m_{A}$. Let
$k_{0}=\lfloor - \log _{10}(\frac{r_{A}}{3})+1\rfloor $.
Let us notice $\mathcal{S}_{X_{A}, k_0} =\left\{ (S_p,S_c,S_m)
\in \mathbb{S}_\mathsf{N}\times (\mathbb{S}_\mathsf{P})^2/ \forall k\leqslant k_{0},\right.$
$\left.(S_p^{k}=(S_p)_{A}^{k})\wedge
(S_c^{k}=(S_c)_{A}^{k})\wedge (S_m^{k}=(S_m)_{A}^{k}))\right\}$.

Then $\forall (S_p,S_c,S_m) \in \mathcal{S}_{X_{A}, k_0},
(S_p,\widetilde{x},S_c,\widetilde{m},S_m)\in \mathcal{B}_{A}.$

Let $\mathcal{J} =\left\{i\in
\llbracket0,\mathsf{N-1}\rrbracket/\check{x}_{i}\neq
\mathcal{X}(X_{B})_{i}, \text{ where }\right.$\newline
$\left.(\check{S_p},\check{x},\check{S_c},\check{m},\check{S_m})=\mathcal{G}_{f_0}^{k_{0}}(X_{A})\right\},$
$\lambda = card(\mathcal{J})$,\newline and
$j_0 <j_1 < ... < j_{\lambda-1}$ the elements of $ \mathcal{J}$.

\begin{enumerate}
  \item Let us firstly build three strategies: $S_p^\ast$,  $S_c^\ast$, and 
  $S_m^\ast$ as follows.
\begin{enumerate}
  \item 
$(S_p^\ast)^k=(S_p)_A^k$, $(S_c^\ast)^k=(S_c)_A^k$, and
$(S_m^\ast)^k=(S_m)_A^k$, if $k \leqslant k_0.$
\item Let us now explain how to replace $\mathcal{X}(X_{B})_{j_q}$, $\forall q \in
  \llbracket0;\mathsf{\lambda-1}\rrbracket$:\newline
First of all, we must replace $\mathcal{X}(X_{B})_{j_0}$:
\begin{enumerate}
  \item If $\exists \lambda_0 \in \llbracket0;\mathsf{P-1}\rrbracket /
  \mathcal{X}(X_{B})_{j_0}=\check{m}_{\lambda_0}$, then we can choose
 $(S_p^\ast)^{k_0+1}=j_0$, $(S_c^\ast)^{k_0+1}=\lambda_0$,
 $(S_m^\ast)^{k_0+1}=\lambda_0$, and so $I_{j_0}$ will be equal to $1$.
 \item If such a $\lambda_0$ does not exist, we choose:\newline
 $(S_p^\ast)^{k_0+1}=j_0$, $(S_c^\ast)^{k_0+1}=0$,
 $(S_m^\ast)^{k_0+1}=0$,\newline $(S_p^\ast)^{k_0+2}=j_0$,
 $(S_c^\ast)^{k_0+2}=0$, $(S_m^\ast)^{k_0+2}=0$ \newline and so let us notice
 $I_{j_0}=2$.\newline

All of the $\mathcal{X}(X_{B})_{j_q}$ are replaced similarly.
The other terms of $S_p^\ast$,  $S_c^\ast$, and $S_m^\ast$ are constructed identically, and the values of $I_{j_q}$ are defined on the same way.\newline
Let $\gamma = \sum_{q=0}^{\lambda-1}I_{j_q}$.
\end{enumerate}

\item $(S_p^\ast)^{k}=(S_p^\ast)^{j}$, $(S_c^\ast)^{k}=(S_c^\ast)^{j}$
and $(S_m^\ast)^{k}=(S_m^\ast)^{j}$ where $j\leqslant k_{0}+\gamma$
is satisfying $j\equiv k~[\text{mod } (k_{0}+\gamma)]$, if
$k>k_{0}+\gamma$.
\end{enumerate}
So,$\mathcal{G}_{f_0}^{k_{0}+\gamma}((S_p^\ast,x_A,S_c^\ast,m_A,S_m^\ast))=(S_p^\ast,x_B,S_c^\ast,m,S_m^\ast)$

Let $\mathcal{K} =\left\{ i \in \llbracket0;\mathsf{P-1}\rrbracket / m_i \neq
\mathcal{M}(X_{B})_{i}, \text{ where }\right.$\newline
$\left.(S_p^\ast,x_B,S_c^\ast,m,S_m^\ast)=\mathcal{G}_{f_0}^{k_{0}+\gamma}((S_p^\ast,x_A,S_c^\ast,m_A,S_m^\ast))\right\},$\newline
$\mu = card(\mathcal{K})$ and $r_0 <r_1 < ... < r_{\mu-1}$ the elements of $
\mathcal{K}$.

  \item Let us secondly build three other strategies:  $\widetilde{S_p}$, 
  $\widetilde{S_c}$, $\widetilde{S_m}$ as follows.
\begin{enumerate}
  \item $\widetilde{S_p}^k=(S_p^\ast)^k$, $\widetilde{S_c}^k=(S_c^\ast)^k$, and
$\widetilde{S_m}^k=(S_m^\ast)^k$, if $k \leqslant k_0 + \gamma.$
\item  Let us now explain how to replace $\mathcal{M}(X_{B})_{r_q}, \forall q
\in \llbracket0;\mathsf{\mu-1}\rrbracket$:
  
First of all, we must replace $\mathcal{M}(X_{B})_{r_0}$:
\begin{enumerate}
  \item If $\exists \mu_0 \in \llbracket0;\mathsf{N-1}\rrbracket /
  \mathcal{M}(X_{B})_{r_0}=(x_B)_{\mu_0}$, then we can choose 
 $\widetilde{S_p}^{k_0+\gamma +1}=\mu_0$, $\widetilde{S_c}^{k_0+\gamma
 +1}=r_0$, $\widetilde{S_m}^{k_0+\gamma
 +1}=r_0$, and $J_{r_0}$ will be equal to $1$.
 \item If such a $\mu_0$ does not exist, we choose:
$\widetilde{S_p}^{k_0+\gamma +1}=0$, $\widetilde{S_c}^{k_0+\gamma
 +1}=r_0$, $\widetilde{S_m}^{k_0+\gamma
 +1}=r_0$,\newline 
$\widetilde{S_p}^{k_0+\gamma +2}=0$, $\widetilde{S_c}^{k_0+\gamma
 +2}=r_0$, $\widetilde{S_m}^{k_0+\gamma
 +2}=0$,\newline 
$\widetilde{S_p}^{k_0+\gamma +3}=0$, $\widetilde{S_c}^{k_0+\gamma
 +3}=r_0$, $\widetilde{S_m}^{k_0+\gamma
 +3}=0$,\newline 
 and so let us notice $J_{r_0}=3$.\newline

All the $\mathcal{M}(X_{B})_{r_q}$ are replaced similarly.
The other terms of $\widetilde{S_p}$, 
$\widetilde{S_c}$, and $\widetilde{S_m}$ are constructed identically, and the
values of $J_{r_q}$ are defined on the same way.\newline
Let $\alpha = \sum_{q=0}^{\mu-1}J_{r_q}$.
\end{enumerate}
\item $\forall k \in \mathbb{N}^\ast, \widetilde{S_p}^{k_0+ \gamma + \alpha
+ k}=(S_p)_B^{k}$, $\widetilde{S_c}^{k_0+ \gamma + \alpha + k}=(S_c)_B^{k}$, and
$\widetilde{S_m}^{k_0+ \gamma + \alpha + k}=(S_m)_B^{k}$.
\end{enumerate}
\end{enumerate}

So,
$\mathcal{G}_{f_0}^{k_{0}+\gamma+\alpha}(\widetilde{S_p},x_A,\widetilde{S_c},m_A,\widetilde{S_m})=X_B$,
with $(\widetilde{S_p},\widetilde{S_c},\widetilde{S_m}) \in \mathcal{S}_{X_{A},
k_0}$.
Then $\widetilde{X} = (\widetilde{S_p},x_A,\widetilde{S_c},m_A,\widetilde{S_m})
\in \mathcal{X}_2$ is such that $\widetilde{X} \in \mathcal{B}_A$ and
$\mathcal{G}_{f_0}^{k_0+\gamma+\alpha}( \widetilde{X}) \in \mathcal{B}_B$.
Finally we have proven the result.
\end{proof}

\subsection{Sensitivity on Initial
Conditions}\label{sec:sensibilite}

\begin{proposition}%[Sensitivity of $,\mathcal{G}_{f_0}$]
\label{prop:sensibilite}
$(\mathcal{X}_2,\mathcal{G}_{f_0})$ has sensitive dependence on initial conditions.
\end{proposition}

\begin{proof}
$\mathcal{G}_{f_0}$ is regular and transitive. Due to Theorem~\ref{theo:banks}, $\mathcal{G}_{f_0}$ is sensitive.
\end{proof}

%In a following  section, we will show that this result can be stated independently of regularity,
%which must be redefined in the context of the finite set of
%machine numbers (see section \ref{Concerning}).

%In addition, the constant of sensitivity will be computed.

%This sensitive dependence could be stated as a consequence of regularity and
%transitivity (by using the theorem of Banks~\cite{Banks92}). However, we
%prefer to prove this result independently of regularity, because the
%notion of regularity must be redefined in the context of the finite set of
%machine numbers (see section \ref{Concerning}).

%In addition, the constant of sensitivity has been computed.

\subsection{Devaney's topological chaos}

In conclusion, $(\mathcal{X}_2,\mathcal{G}_{f_0})$ is topologically transitive, regular,
and has sensitive dependence on initial conditions. Then we have the result.

\begin{theorem}%[$\mathcal{G}_{f_0}$ is a chaotic map]
\label{theo:chaotic}
$\mathcal{G}_{f_0}$ is a chaotic map on $(\mathcal{X}_2,d_2)$ in the sense of Devaney.
\end{theorem}

So we can claim that:

\begin{theorem}%[Chaos-security of SCISMM]
\label{theo:SCISMM-chaosecurity}
SCISMM is topologically-secure.
\end{theorem}

\section{Conclusion}\label{sec:conclusion}

\noindent In this research work, a new information hiding scheme has been
introduced. It is topologically-secure and stego-secure, and thus is able to
withstand attacks in Watermark-Only Attack (WOA) and Constant-Message Attack (CMA) setups. These results have been obtained after having studied the topological behavior of this data hiding scheme.
To the best of our knowledge, this algorithm is the third scheme that has been proven to be secure, according to the information hiding security field.

In future work, we intend to study the robustness of this scheme, and to compare it with the two other secure algorithms.
Additionally, we will investigate the topological properties of our scheme, to see whether it is secure in KOA and KMA setups.

\vfill
\bibliographystyle{plain}
\bibliography{jabref}

% \section*{\uppercase{Appendix}}
% 
% \noindent If any, the appendix should appear directly after the
% references without numbering, and not on a new page. To do so please use the following command:
% \textit{$\backslash$section*\{APPENDIX\}}
% 
% \vfill
\end{document}